 \documentclass[pra,twocolumn,showpacs,showkeys]{revtex4-1}

\usepackage[active]{srcltx}
\usepackage[english]{babel}
\usepackage[T1]{fontenc}
\usepackage{epsfig}
\usepackage{tipa}
\usepackage{bbold}
\usepackage{graphicx}
\usepackage{amsthm,mathrsfs}
\usepackage{amssymb,bm,mathtools}
\usepackage{hyperref}
\usepackage[dvipsnames]{xcolor}
\usepackage{comment}
\usepackage{enumerate}
\usepackage{framed}
\definecolor{shadecolor}{rgb}{.8,.8,.8}

\def\Z{\mathbb{Z}}
\def\>{\rangle}
\def\<{\langle}

\def\ie{{i.e.} }
\renewcommand{\v}[1]{\ensuremath{\mathbf{#1}}}
 
\newcommand{\ket}[1]{\left| #1 \right>}
\newcommand{\bra}[1]{\left< #1 \right|}

\newtheorem*{theo*}{Theorem}
\newtheorem*{lemma*}{Lemma}
\newtheorem{lemma}{Lemma}
\newtheorem*{remark*}{Remark}

\newtheorem*{dfn*}{Definition}
\newtheorem*{prp*}{Proposition}
\newtheorem{prp}{Proposition}
\newtheorem*{conj*}{Conjecture}
\hypersetup{colorlinks=true,linkcolor=RoyalBlue}

\begin{document}
\title{Quantum walks with a one-dimensional coin}

\author{Alessandro Bisio}
\author{Giacomo Mauro D'Ariano}
%\email[]{}
\author{Marco Erba}
\author{Paolo Perinotti}
\author{Alessandro Tosini}
\affiliation{Universit\`a degli Studi di Pavia, Dipartimento di Fisica, QUIT Group, and INFN Gruppo IV, Sezione di Pavia, via Bassi 6, 27100 Pavia, Italy}

\begin{abstract}
  Quantum walks (QWs) describe particles evolving coherently on a
  graph. The internal degree of freedom corresponds to a Hilbert
  space, called \emph{coin system}. We consider QWs on Cayley graphs
  of some group $G$. In the literature, investigations concerning
  infinite $G$ have been focused on graphs corresponding to
  $G=\mathbb{Z}^d$ with coin system of dimension 2, whereas for
  one-dimensional coin (so called \emph{scalar} QWs) only the case of
  finite $G$ has been studied.  Here we prove that the evolution of a
  scalar QW with $G$ infinite Abelian is trivial, providing a thorough
  classification of this kind of walks. Then we consider the infinite
  dihedral group $D_\infty$, that is the unique non-Abelian group $G$
  containing a subgroup $H\cong\mathbb{Z}$ with two cosets. We
  characterize the class of QWs on the Cayley graphs of $D_\infty$
  and, via a coarse-graining technique, we show that it coincides with
  the class of spinorial walks on $\mathbb{Z}$ which satisfies parity
  symmetry. This class of QWs includes the \emph{Weyl} and the
  \emph{Dirac} QWs. Remarkably, there exist also spinorial walks that
  are not coarse-graining of a scalar QW, such as the \emph{Hadamard
    walk}.
\end{abstract}
\maketitle

\section{Introduction}
Quantum walks (QWs) are the quantum version of the \emph{classical
  random walks}, which made their first appearance in physics with
Einstein's seminal work on Brownian motion \cite{E05}. A peculiarity
of QWs on graphs with respect to their classical counterpart is that
the vertexes of the graph carry an internal degree of freedom (spin,
helicity, etc.) corresponding to a finite-dimension Hilbert space,
called \emph{coin system}.

Models of QWs have been broadly studied in diverse formulations
\cite{ADZ93,M96,AB01,S03,KRS04,M07}, since they revealed to be
suitable both as a simulation tool---e.g. in \emph{lattice gauge
  theories} \cite{K83,CNK84,AD16}---and as a computational
one---e.g. in designing \emph{quantum algorithms}
\cite{CCD03,A04,MFS05}.  Recently, QWs have also been exploited as
discrete models of spacetime
\cite{BDT15,DP14,ArrighiNJP14,BBDPT15,BDP15,AFF15}. Discrete-time QWs on lattices have been
studied in the continuum limit in
Refs. \cite{BB94,DP14,FS14,BDP16}, recovering Weyl, Dirac and
Maxwell dynamics.

In this manuscript we consider discrete-time QWs on an infinite graph
requiring \emph{locality} and \emph{homogeneity} of the evolution. The
former implies that each site (each vertex of the graph) has a finite
number of first neighbours, while the latter means the
indistinguishability of the sites based on the evolution in a sense
formalized in Ref. \cite{DP14}, where it is proved that these
hypotheses amount to require the graph to be a \emph{Cayley graph} of
a finitely generated group $G$. Cayley graphs, as diagrammatic
counterparts of groups, are convenient means to study QWs exploiting
the group-theoretical machinery.

In the case of Abelian $G$, one can represent the walk in the wave-vectors space via the Fourier transform, resorting to the (one-dimensional) irreducible representations of $G$. This allows one to diagonalize the walk evolution operator and to simply solve the walk dynamics in terms of its \emph{dispersion relations}.

The procedure is not so straightforward in the non-Abelian case. Indeed, a method working in the general case is still lacking, due to the fact that representations of the infinite discrete non-Abelian groups are generally unknown. In Ref. \cite{DEPT15}, a novel technique allowing to tackle this issue in the case of \emph{virtually Abelian groups} has been presented.

A virtually Abelian group $G$ is generally a non-Abelian group with an
Abelian subgroup $H$ of finite \emph{index} (the number of cosets of
$H$ in $G$). This property enables one to define a notion of
wave-vector as an invariant of the dynamics also in the non-Abelian
case, thus solving the dynamics of these particular non-Abelian
QWs. This technique can be viewed as a \emph{coarse-graining} of the
walk: the original virtually Abelian QW on $G$ is unitarily equivalent
to a walk on $H$ with a larger coin system.  

We will apply this method to the most elementary case of \emph{scalar}
QWs, i.e. walks with a one-dimensional coin system. This kind of walks
are the most elementary ones from the point of view of the coin
system, but unfortunately this does not mean that they are the easiest
to treat, since their existence imposes additional constraints on the
graph (see Sec. \ref{cayley}). Scalar QWs on Cayley graph have been
explored in Ref. \citep{ARC08}, where the authors restricted the
investigation to finite groups, classifying scalar QWs on Cayley
graphs with two and three generators. The present investigation
focuses on infinite groups. The framework of scalar QWs differs from
that of staggered QWs,  \ie QWs without coin tossing, recently
considered in the literature
\cite{PhysRevA.71.032347,PhysRevA.91.052319,Santos:2015:MCQ:2822142.2822148}
where the walk is defined by an evolution operator that is the product
of two reflections acting on the site basis.

After reviewing QWs on Cayley graphs, in Sec. \ref{abelian} we first investigate and classify
infinite Abelian scalar QWs, extending the results of Ref. \citep{MD96} to any infinite Abelian
group, and with arbitrary presentations. Then, in Sec. \ref{group-der}, we consider the simplest
case of non-Abelian group $G$ with a subgroup $H \cong \mathbb{Z}$ of index 2. Such a group is the
infinite dihedral group $D_{\infty}$, and we derive all its Cayley graphs admitting a scalar QW with
a coarse-grained scheme having coordination number 2. All the scalar QWs on these Cayley graphs are
derived in Sec. \ref{dihedral-qw}.  We show that their coarse-graining coincides (up to a local
change of basis) with the class of QWs on $\ell^2(\mathbb{Z})\otimes \mathbb{C}^2$ that are
invariant under parity transformation. The walks in this class are studied via the usual
Fourier-transform method and it results that they can exhibit both linear and massive dispersion
relations.

\section{QWs on Cayley graphs}\label{cayley}
In this section we review the notion of quantum walks on Cayley graphs, previously discussed in Refs. \cite{AG06,ARC08,KP09}.

Let $G$ be a group: we can always select a \emph{generating set} $S_+$ for $G$, namely a subset $S_+
\subseteq G$ such that any element of the group can be built as composition of elements $g\in S_+$
and their inverses. In the following we do not assume the generating set to be \emph{symmetric} (a
generating set $S_+$ is called symmetric if $S_+=S_-$, where $S_-$ is the set of inverses
of the elements in $S_+$). In order to specify a group, as well as a generating set $S_+$, a set $R$
of \emph{relators} is also needed, namely some \emph{words} formally built by composition of
elements $g\in S_+$ and corresponding to the identity element $e\in G$. For example, if $R$ is
trivial, one gets the free group on $S_+$.

These two ingredients provide a so-called \emph{presentation} $G=\<S_+|R\>$ of a group. Presentations are not in one-to-one correspondance with groups: given a group $G$, it has in general different presentations. However, any presentation completely specifies a unique group.

Presentations of groups have a convenient geometrical representation: \emph{Cayley graphs}. Given a
group $G$ and a generating set $S_+$ for $G$, the Cayley graph $\Gamma (G, S_+)$ is defined as the
edge-colored directed graph having vertex set $G$, edge set $\{(x, xg); x \in G, g \in S_+\}$, and a
color assigned to each generator $g \in S_+$. Besides, an edge corresponding to a generator $g\in
S_+$ is usually represented as undirected when $g^2=e$. Cayley graphs are indeed in one-to-one
correspondence with presentations. Relators are just closed paths over the graph, i.e.
\emph{cycles}, and conversely any cycle on the graph can be built as composition of some relators.

In the following we will consider Cayley graphs of finitely generated
groups ($|S_+|< \infty$) as the graphs of our quantum walks. A
discrete-time quantum walk on a Cayley graph $\Gamma(G,S_+)$ with an
$s$-dimensional coin system ($s\geq 1$) is a unitary evolution of a
system with Hilbert space $\ell^2(G)\otimes\mathbb{C}^s$ such that
\begin{align}|\psi_{g,t+1}\>=\sum_{h\in S_+}A_{h}|\psi_{gh,t}\>,\label{sumA}
\end{align}
where $0\neq A_h\in M_s(\mathbb{C})$ are the {\em transition matrices} of the walk.  In the
following, we will consider $S_+$ generally non symmetric. In the previous literature \cite{DP14}
the sum in Eq. (\ref{sumA}) was extended to $S_+\cup S_-$. For this reason, for the sake of
uniformity, we will explicitly name the walk {\em monoidal} whenever $S_+$ is not symmetric.
\footnote{The motivation of keeping all matrices nonvanishing originates in Ref. \cite{DP14} from
  the logic of deriving the graph inversely from a set of nonnull matrices. This is relavant from a
  derivation of the QW (more generally quantum automaton) from general topological principles of a
  countable set of interacting systems.}

Considering the right regular
representation $T_g$ of the group $G$, whose action on $\ell ^2(G)$ is defined as $T_g|x\>\coloneqq
|xg^{-1}\>$, we can represent the QW through
\begin{align*}
A\coloneqq \sum_{h\in S_+}T_{h}\otimes A_h.  
\end{align*}
The unitarity conditions for the walk operator $A$ are
$AA^{\dagger}=A^{\dagger}A=T_e\otimes I_s$; these conditions, for a scalar QW of the form
\[
A\coloneqq\sum_{h\in S_+}T_{h} z_h
\]
(where the $z_h\in \mathbb{C}$ are called \emph{transition scalars}), lead to the set of equations:
\begin{equation}\label{unit}
\sum_{\substack{hh'^{-1}=g \\ h\neq h'}} z_hz_{h'}^* = 0,\ \ \sum_{\substack{h^{-1}h'=g \\ h\neq h'}} z_{h}^*z_{h'} = 0,\ \ \sum_h |z_h|^2 =1.
\end{equation}
It is immediate to check that trivial solutions of Eq. (\ref{unit}) with $A=T_h$ can occur only for
monoidal walks QW with singleton $S_+=\{h\}$.  A necessary condition for the existence of solutions
of Eq. (\ref{unit}) is given by the following Lemma.
\begin{lemma}[]\label{quadr}
  Given a Cayley graph $\Gamma(G,S_+)$, a necessary condition for the existence of a scalar
  quantum walk $A = \sum_{h\in S_+}T_hz_h$ on $\Gamma(G,S_+)$ is that, for each ordered pair
  $(h_1,h_2)\in S_+\times S_+$ such that $h_1 \neq h_2$, there exists at least a different pair
  $(h_3,h_4)$ such that $h_1h_2^{-1}=h_3h_4^{-1}$. This is called \emph{quadrangularity condition}.
  \cite{ARC08}
\end{lemma}

\subsection{Free Abelian QWs}

The case of QWs on Cayley graphs of a free Abelian group---i.e. $G\cong \Z^d$---is the simplest to treat in order to analytically solve the dynamics, since the walk can be easily diagonalized by a Fourier transform. We will label the elements
$\v{x}\in\Z^d$, using the additive notation for the group composition. The right regular representation is decomposed into one-dimensional irreducible representations, since the group is Abelian. One can thus diagonalize $T_{\v{x}}$ in the wave-vector space as follows
\begin{equation*}
\ket{\v{k}} \coloneqq \frac{1}{\left(2\pi\right)^{\frac{d}{2}}}\sum\limits_{\v{x'}\in \Z^d} e^{-i\v{k}\cdot \v{x'}}\ket{\v{x'}},\quad T_{\v x} \ket{\v{k}} = e^{-i \v{k}\cdot \v{x}} \ket{\v k},
\end{equation*}
where $\v k$ belongs to the {\em first Brillouin zone} $\mathcal{B}\subseteq \mathbb R^d$, which is the largest set 
that contains vectors $\v k$ corresponding to inequivalent elements $\ket{\v k}$.
The evolution operator of the walk then reads
\begin{align*}
A= \int_{\mathcal{B}} d\v{k} \ket{\v{k}}\bra{\v{k}}\otimes A_{\v{k}},\quad A_{\v{k}} \coloneqq \sum_{\v{h}\in S_+} e^{-i\v{k}\cdot \v{h}}A_{\v{h}},
\end{align*}
where $A_{\v{k}}$ is unitary $\forall \v{k}\in \mathcal{B}$.  Being $A_{\v{k}}$ unitary, the
eigenvalues are phase factors of the form $e^{i \omega_r (\v{k})}$: the collection $\{\omega_r
(\v{k})\}_{r=1,\ldots , s}$ for $\v k \in \mathcal{B}$ are called the \emph{dispersion relations} of
the QW, and give the kinematics of the walk. Its first and second derivatives, indeed, provide
respectively the \emph{group velocity} and the \emph{diffusion coefficient} of particle states.

\subsection{Coarse-graining of QWs}
Our aim is to study scalar QWs on some group $G$ containing $H\cong \mathbb{Z}$ as a subgroup, with
finitely many cosets in $G$. The minimal choice is $G=H\cup Hr$, where $r$ is a coset
representative. The group $G$ is then virtually Abelian by definition. As a consequence, one can
apply the coarse-graining procedure presented in Ref.  \cite{DEPT15} and study the kinematics of the
walk in the $k$-space, likewise in the purely Abelian case.

This technique is applied through a unitary transformation on the walk operator: it is nothing but a
change of representation of the generators of $G$, allowing one to represent the QW on $G$ as a
coarse-grained QW on $H$ having larger coin system. In particular, two different choices of the subgroup $H$ do not
change the dispersion relations, which are informative about the kinematics of the system. 

The core idea is to choose a partition of $G$ into cosets of $H$, assigning to them a finite set of
labels. The vertexes of the original Cayley graph of $G$ are grouped into clusters---containing one
vertex from each coset---which become the vertexes of the new coarse-grained walk on $H$. The coset
labels designate now an additional internal degree of freedom.

A virtually Abelian quantum walk on $\ell ^2(G)\otimes\mathbb C^s $ can be regarded as an Abelian QW
on $\ell ^2(H)\otimes\mathbb C^{s\times l}$, where $l$ is the index of $H$ in $G$. In the present
case, $s=1$ and $l=2$. The coarse-graining procedure is performed by choosing a \emph{regular
  tiling}, namely a particular coset partition of $G$ with respect to an Abelian subgroup $H$ of
finite index. Accordingly, we will choose $G=Hc_1\cup Hc_2$, with $H\cong \mathbb{Z}$ and $c_1,c_2$
arbitrary coset representatives.

We will denote the generators of $G$ by $h \in S_+$. Having chosen the coset representatives $\{c_j\}_{j=1,2}$, we can define a unitary mapping between $\ell^2\left(G\right)$ and $\ell^2\left(H\right)\otimes\mathbb{C}^2$ as follows
\begin{align*}
U_H:\ell^2\left(G\right) &\rightarrow\ell^2\left(H\right)\otimes\mathbb{C}^2;\; U_H\ket{xc_j}=\ket{x}\ket{j},\; \forall x\in H, 
\end{align*}
for $j=1,2$.  In Ref. \cite{DEPT15} it is shown that---since $\forall x\in H$, $\forall h\in S_+$ and
$\forall c_j$ there exist $x' \in H$ and $j'=\tau(h,j)\in\{1,2\}$ such that $xc_jh^{-1}=x'c_{j'}$---the
coarse-grained generating set $\tilde{S}_+=\{\tilde h\}\subseteq H$ is defined as
\begin{equation}\label{coarse-gen}
\tilde{S}_+\coloneqq\{c_{\tau(h,j)}hc^{-1}_j|h\in S,j=1,2\},
\end{equation}
while their corresponding transition matrices will be given by
\begin{equation}\label{coarse-matr}
(A_{\tilde{h}})_{ij} = \sum_{h\in S_+} z_h \delta_{\tilde{h},c_i hc_{j}^{-1}}\delta_{i,\tau(h,j)}.
\end{equation}
Finally, the coarse-grained evolution operator reads
\[
\begin{split}
\mathcal{R}[A] &= \left( U_H\otimes \mathbb{1} \right) A \left( U_H\otimes \mathbb{1} \right)^{\dagger} = \\
& = \sum_{h\in S_+}\sum_{j=1,2} T_{c_{\tau(h,j)}hc_{j}^{-1}}\otimes  \ket{\tau(h,j)} \bra{j} z_h,
\end{split}
\]
where clearly now $T$ is the right regular representation of $H$. We say that a QW $A$ is a
\emph{coarse-grained scalar QW} if there exists a scalar walk $A'$ such that $A=\mathcal{R}[A']$.

It is known \cite{M96} that the only scalar QWs on $\mathbb{Z}$ are the monoidal QW $A_\pm: =
e^{-i\theta_\pm} T_\pm$, with $\theta_{\pm}$ arbitrary phases, and $T_\pm$ the right/left shift
operators on $\ell^2\left( \mathbb{Z} \right)$. Here we analyze scalar QWs on a group that is
virtually Abelian with Abelian subgroup $H\cong\mathbb{Z}$, starting from the easiest case in which
the index of $H$ is $2$.  Furthermore, we require the coarse-grained QW on
$\ell^2(\mathbb{Z})\otimes \mathbb{C}^2$ to have coordination number two, which restricts the class
of Cayley graphs of $G$ that we will study. Our analysis will lead to a classification of all the
spinorial QW with coordination number 2 on $\ell^2(\mathbb{Z})\otimes \mathbb{C}^2$ that can be
obtained as a coarse-grained scalar QW.

The coarse-graining technique will be applied in Secs. \ref{group-der}
and \ref{dihedral-qw}. In the following Section we classify infinite
Abelian scalar QWs.

\section{Classification of infinite Abelian scalar QWs on Cayley graphs}\label{abelian}
By the \emph{fundamental theorem of finitely generated Abelian groups}, the generic infinite group
of this kind is of the form $G=\mathbb{Z}_{i_1}\times \ldots \times \mathbb{Z}_{i_n}\times
\mathbb{Z}^d$, for $d \geq 1$ and $0\leq n<\infty$. We now give a full characterization of infinite Abelian scalar QWs,
providing a general structure for the evolution operator of these walks in the following Proposition.
\begin{prp}\label{class}
  Let $A$ be the unitary operator of a scalar QW on the Cayley graph of $G=\mathbb{Z}_{i_1}\times
  \ldots \times \mathbb{Z}_{i_n}\times \mathbb{Z}^d$ for $1\leq d<\infty$, and $0\leq n<\infty$.
  Then $A$ splits into the direct sum of one-dimensional monoidal QWs $e^{-i\theta_j}T_j$, with
  $T_j$ shift operators over $\mathbb{Z}^d$, and $j\in\{ 1,2,\ldots,i_1\times i_2\times \ldots\times
  i_n \}$. In particular, the dispersion relations are linear in the wave-vectors.
\end{prp}
\begin{proof}
  Let us pose $G=\mathbb{Z}_{i_1}\times \ldots \times \mathbb{Z}_{i_n}\times \mathbb{Z}^d
  =\mathrel{\mathop:} F\times \mathbb{Z}^d$, for $d \geq 1$. We can decompose the elements of $G$
  into one component in $F$ and one in $\mathbb{Z}^d$; accordingly, for some $h\in S_+ \subset
  \mathbb{Z}^d$, we define $R(h) \subseteq F$ such that $\bigcup_{h\in S_+}R(h)\times\{ h\}$ is a set
  of generators for $G$. Let $C_l$ be the right regular representation of the generator of
  $\mathbb{Z}_{i_l}$ \footnote{Also in the finite Abelian case one can decompose the right regular
    representations into irreducible representations: the $k$-space is discrete and the
    diagonalization reads $C_l = \sum_{j=1}^{i_l} \ket{j}\bra{j} e^{2\pi i\left( \frac{j}{i_l}
      \right)}$, where the wave-vectors are the $2\pi j/i_l$.}. Then, defining for any $f\in F$ the
  integers $m_l(f)\in \lbrace 1,\ldots , i_l \rbrace$ such that $T_f = C_1^{m_1(f)} \otimes \ldots
  \otimes C_n^{m_n(f)}$, the diagonalization of the general QW on $\ell^2(G)\otimes \mathbb{C}$
  reads
\begin{align}
\begin{split}\label{abelian-qw}
  A & = \sum_{h\in S_+} \sum_{f\in R(h)} \left( T_{f} \otimes T_{h} \right) z_{(f,h)} = \\
  & = \sum_{j_1=1}^{i_1} \ldots \sum_{j_n=1}^{i_n} \ket{j_1}\bra{j_1} \otimes \ldots \otimes
  \ket{j_n}\bra{j_n} \otimes \sum_{h\in S_+} T_{h} z_h(\v{j}),
\end{split}
\end{align}
where
\begin{equation*}
z_h(\v{j})\coloneqq \sum_{f\in R(h)} z_{(f,h)} e^{2\pi i\left(
        j_1\frac{m_1(f)}{i_1} + \ldots + j_n\frac{m_n(f)}{i_n} \right)},
\end{equation*}
and $\v{j} \coloneqq (j_1,\ldots , j_n)$. The evolution operator \eqref{abelian-qw} is now
block-diagonalizable in the $k$-space as
\[
\int_{\mathcal{B}}\!\! \mathrm{d}\v{k}\; \ket{\v{k}}\!\bra{\v{k}} \left( \sum_{\v{h}\in S_+} e^{-i\v{k}\cdot \v{h}}z_\v{h}(\v{j}) \right)\,, \forall \v{j},
\]
% = e^{i\omega(\v{k},\v{j})}$ for unitarity. Then one %has, for all fixed $\v{j}$,
%\[
%\sum_{\v{h},\v{h}'} z_\v{h}(\v{j})z_{\v{h}'}^*(\v{j})e^{i \v{k} \cdot \left( \v{h}-\v{h}' \right)} = 1\ \ \forall \v{k}\in \mathcal{B},
%\]
with $A_{\v{k}}(\v{j}) \coloneqq  \sum_\v{h} e^{-i\v{k}\cdot \v{h}}z_\v{h}(\v{j})$ unitary by construction. This leads to the unitarity conditions \eqref{unit}. Take now, for $\v{h}_i,\v{h}_j\in S_+$, the collection $M$ of all the $\v{v}= \v{h}_i-\v{h}_j\in \mathbb{Z}^d$ such that
\begin{equation}\label{def-v}
\lVert \v{v} \rVert = \max_{\v{h}_i,\v{h}_j\in S_+}\lVert\v{h}_i-\v{h}_j\rVert.
\end{equation}
Suppose that, for some $\v{v}'\in M$, there exist two distinct pairs such that
\begin{equation}\label{v}
\v{v}'=\v{h}_1-\v{h}_2=\v{h}_3-\v{h}_4,
\end{equation}
with $\v{v}'\neq 0$ (otherwise $d=0$). Then let's define $\v{d}_{ij} \coloneqq  \v{h}_i-\v{h}_j$: by
definition one has $2\v{v}' = \v{d}_{14}+\v{d}_{32}$ and 
\[
\begin{split}
2\lVert\v{v}'\rVert & = \lVert\v{d}_{14}+\v{d}_{32}\rVert \leq \lVert\v{d}_{14}\rVert+\lVert\v{d}_{32}\rVert \leq \\
&\leq \lVert\v{v}'\rVert + \lVert\v{v}'\rVert = 2\lVert\v{v}'\rVert,
\end{split}
\]
where we used the triangle inequality and the definition \eqref{def-v} of $\v{v}'$. This implies
that $\v{d}_{14}\propto \v{d}_{32}$ and
$\lVert\v{v}'\rVert=\lVert\v{d}_{14}\rVert=\lVert\v{d}_{32}\rVert$, which in turn imply
$\v{d}_{14}=\pm\v{d}_{32}$; this, combined with \eqref{v} and $\v{v}'\neq0$ finally gives
\[
\v{h}_1=\v{h}_3,\ \v{h}_2=\v{h}_4,
\]
i.e. the pair is unique. Then, by the unitarity conditions \eqref{unit}, one has $z_{\v{h}_1}(\v{j})z_{\v{h}_2}^*(\v{j})=0$, namely e.g. $z_{\v{h}_1}(\v{j})=0$. The above argument can thus be iterated removing $\v{h}_1$ from the set $S_+$: one finally concludes that, for each $\v{j}$, just one transition scalar $z_{\tilde{\v{h}}(\v{j})}(\v{j})$ is nonvanishing, namely
\begin{equation*}
A_{\v{k}}(\v{j}) = \sum_{\v{h}\in S_+} \delta_{\tilde{\v{h}}(\v{j}),\v{h}} z_\v{h}(\v{j})e^{-i\v{k}\cdot \v{h}} = z_{\tilde{\v{h}}(\v{j})}(\v{j})e^{-i\v{k}\cdot \tilde{\v{h}}(\v{j})},
\end{equation*}
with $z_{\tilde{\v{h}}(\v{j})}(\v{j}) \eqqcolon e^{-i\theta_{\tilde{\v{h}}}(\v{j})}$ arbitrary phase factor
(by unitarity of $A_{\v{k}}(\v{j})$). We now conveniently define $\v{h}_j \coloneqq
\tilde{\v{h}}(\v{j}),\theta_j \coloneqq \theta_{\tilde{\v{h}}}(\v{j})$ and substitute in
\eqref{abelian-qw}. Thus we finally conclude that any infinite Abelian scalar QW is given by the
direct sum of scalar walks on $\mathbb{Z}$, namely
\begin{equation*}
A=\bigoplus_{j\in I}e^{-i\theta_j}T_j,
\end{equation*}
or else, in the Fourier representation
\[
A =  \int_{\mathcal{B}} \!\!\mathrm{d}\v{k} \; \bigoplus_{j\in I} e^{-i (\v{k} \cdot \v{h}_{j} + \theta_{j})} \otimes \ket{\v{k}}\!\bra{\v{k}},
\]
for $I=\{ 1,2,\ldots,i_1\times i_2\times \ldots\times i_n \}$ and with $\v{h}_{j}\in S_+\subset
\mathbb{Z}^d$, $\theta_{j}$ arbitrary phases. This finally proves that the dispersion relations are
linear in $\v{k}$.
\end{proof}

Notice that the argument of the proof does not hold in the case of a general scalar QW on a finite
Abelian group $F$. For example, it's easy to verify that the Cayley graph corresponding to
$\mathbb{Z}_2 \times \mathbb{Z}_2 = \langle g_1,g_2 | g_1^2,g_2^2,g_1g_2g_1^{-1}g_2^{-1} \rangle$
(square graph) admits a nontrivial scalar quantum walk.

\section{Scalar QWs on the infinite dihedral group}

\subsection{Classification of the Cayley graphs}\label{group-der}
We provided a full classification of infinite Abelian scalar
QWs. Accordingly, in the following we will consider just non-Abelian
scalar QWs. We now aim to derive all the possible non-Abelian groups
$G \cong \mathbb{Z}\cup \mathbb{Z}r$ and their Cayley graphs
satisfying the quadrangularity condition.

It is easy to show that an index-2 subgroup is always normal: by contradiction, let $G$ be an arbitrary group and $H$ be a index-2 subgroup which is not normal in $G$. Accordingly, for some $x_1,x_2 \in H$, the relation $rx_1r^{-1} = x_2r$ holds, and this implies $rx_1 = x_2r^2$. However $r^2$ must be equal to some $xr$, with $x \in H$, otherwise from the previous equation one would have $r \in H$. On the other hand, $r^2 = xr$ reads $r = x \in H$, which is absurd. Then an index-2 subgroup is always normal: left and right cosets coincide. We conventionally choose right cosets to perform the coarse-graining.

Choosing $G$ to be non-Abelian, let us pose $H = \langle a \rangle$: one has $rar^{-1}=a^m$ for some
integer $m\neq 0,1$ (by normality of $H$); then $a=r^{-1}a^mr=(r^{-1}ar)^m=a^{lm}$, for
$l=\frac{1}{m}$ integer: the only possibility is $m=-1$. Thus we have $rar^{-1}\eqqcolon\varphi(a)=a^{-1}$.
Now we prove that $r^2=e$. Indeed, it must be $r^2\in H$. Let's now suppose that $r^2=a^p$: then one
has $r^{-1}a^p=a^pr^{-1} = r^{-1}a^{-p}$, implying $p = 0$ and finally $r^2 = e$. Accordingly, since
defining $C = \langle r|r^2 \rangle$ one has $G=HC$ and $H\cap C=\lbrace e\rbrace$, it follows
that
\[
G = H \rtimes_{\varphi} C \cong \mathbb{Z} \rtimes_{\varphi} \mathbb{Z}_2 = D_{\infty},
\]
namely the infinite dihedral group, with the inverse map $\varphi$ being the only nontrivial
automorphism of $\mathbb{Z}$ that achieves the semidirect product with $\mathbb{Z}_2$.

Since the cosets are mutually disjoint (they define equivalence
classes), the elements of each coset of $H$ in $G$ define a distinct
subset of vertexes of a Cayley graph of $G$; in fact, the union of
these subsets fills all the vertexes associated to the Cayley graph of
$G$. Each element of $H$ is in one-to-one correspondence with an
element of $Hr$ through elements of the form $a^nr$.
%\begin{center}
%\centering
%\includegraphics[width=.8\linewidth]{index-2.pdf}
%\end{center}

We now derive the admissible Cayley graphs of $D_{\infty}$ satisfying the quadrangularity condition
of Lemma \ref{quadr}. In order to find the coarse-grained generators $\tilde{h}$, one has to
explicitly compute the set $\tilde{S}_+$ in Eq. \eqref{coarse-gen}. The $\tilde{h}$ depend in general on the cosets
representatives (which are arbitrary): accordingly, we shall pose a general form $c_1 = a^{m},c_2 =
a^{m'}r$.

We are interested in walks represented on $\ell^2(\mathbb{Z})\otimes \mathbb{C}^2$ with coordination
number 2. Correspondingly, for the coarse-grained generators, we shall impose the condition
\begin{equation}\label{first-cond}
\tilde{h}\in\{e,a,a^{-1}\}.
\end{equation}
We will then exclude the case $a^{l}\in S_+$ with $|l| \geq 2$, since by \eqref{coarse-gen},
choosing for example $j=1$ one would have $\tilde h=a^l$ which would give rise to coarse-grained
walks with coordination number larger than two. Moreover, we can always include $e$ in $S_+$, since
by \eqref{coarse-gen} the identity element is invariant under coarse-graining.

\textbf{Case $\mathbf{a\in S_+}.$} All the generators beside $a$ and $e$ belong
to the coset $Hr$, and are thus of the form $a^nr$. Moreover,
combining \eqref{coarse-gen} and \eqref{first-cond}, we must have
$|n-(m'-m)|\leq 1$, namely
$h\in\{a^{(m'-m)}r,a^{(m'-m)+1}r,a^{(m'-m)-1}r\}$.  By
quadrangularity, we need some generators $h,h'\in Hr$ such that
$a^2=hh'^{-1}$, implying that $a^{(m'-m)\pm1}r\in S_+$.  Moreover, since
$h'h^{-1}=a^{-2}$ by quadrangularity it is also $a^{-1}\in S_+$. We can
possibly include $a^{(m'-m)}r$ in $S_+$, having $e$ as coarse-grained
generator. However, it is easy to check that $\forall m,m'$ these
choices give rise, topologically, to the same Cayley graph (modulo a
constant left-translation $a^{(m'-m)}$). Here it follows an example
for the choice $m'-m=2$ (one moves between sites horizontally through
$a^{\pm 1}$, while vertically through $r$, which has no associated edges in
this example):
\begin{center}
\includegraphics[width=.8\linewidth]{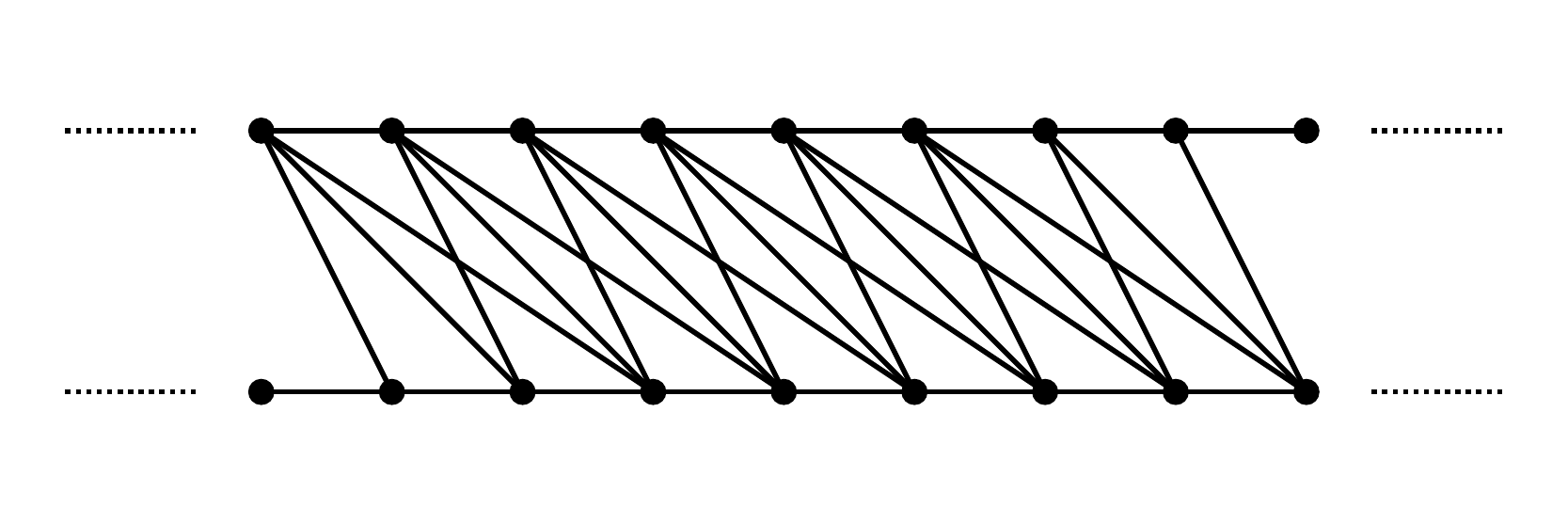}
\end{center}
Accordingly, one can just set $m=m'=0$, namely:
\begin{center}
\includegraphics[width=.8\linewidth]{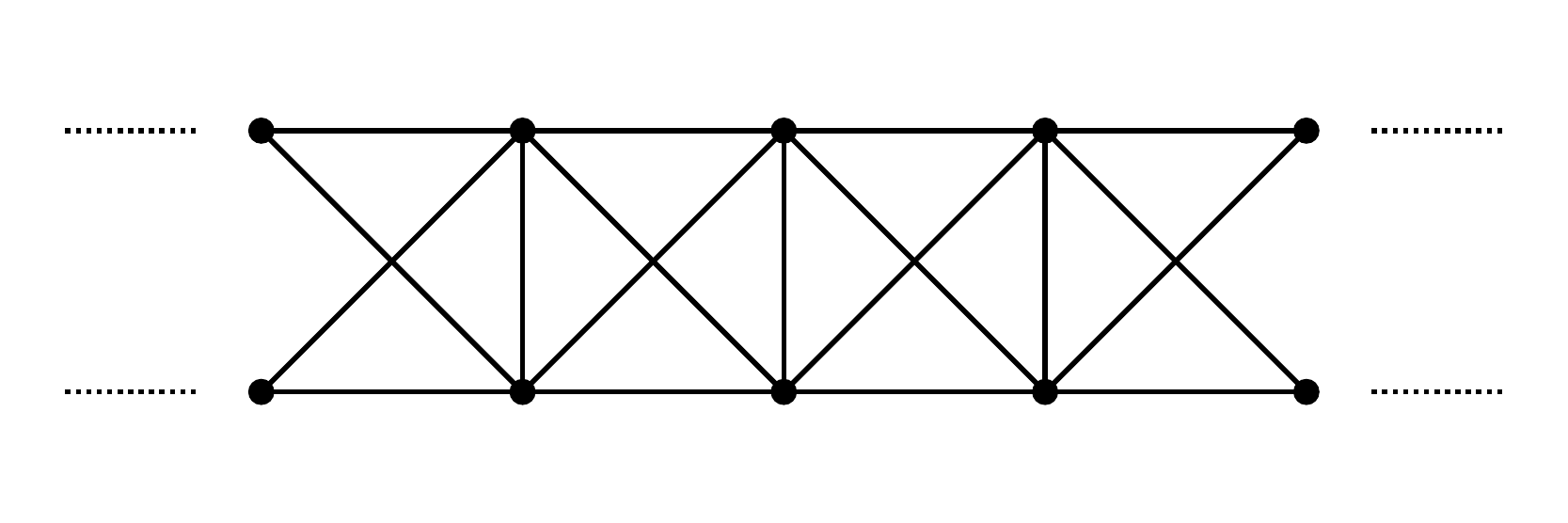}
\end{center}

\textbf{Case $\mathbf{a\not\in S_+}.$} From the previous case we know
that $a\in S_+\Leftrightarrow a^{-1}\in S_+$, then obviously
$a\not\in S_+\Rightarrow a^{-1}\not\in S_+$. Then
$S_+\subseteq \{e,h_i=a^{n+i}r| i=-1,0,1\}$. However, for any pair
$(h_i,h_j)$ with $i\neq j$, there does not exist a different pair $(h,h')$ such that $hh'^{-1}=h_ih_j^{-1}=a^{2}$,
thus violating quadrangularity. This rules out the case $a\not\in S_+$.

Colouring consistently the graph derived, one finds the most general admissionmissible Cayley graph of $D_{\infty}$, which is shown in Fig. \ref{fig:fig2} together with the corresponding presentation.

\begin{figure}
% \includegraphics[width=1\linewidth]{dihedral.pdf}
% \caption{(colors online) The most general Cayley graph of the infinite dihedral group which admits a scalar quantum walk, corresponding to the presentation: $D_{\infty} = \langle a,b,c,d|b^2,c^2,bda^{-1},cda,bab^{-1}a,cac^{-1}a,bca^{-2}\rangle$.
% The generators (here $a^{-1}$ is omitted for graphical simplicity) are associated to edges of the graph, each corresponding to a transition scalar of the walk. One can drop $d$ and the relators containing it. Moreover, one can include $e$ in the generating set, which would correspond to a loop at each site.
\includegraphics[width=1\linewidth]{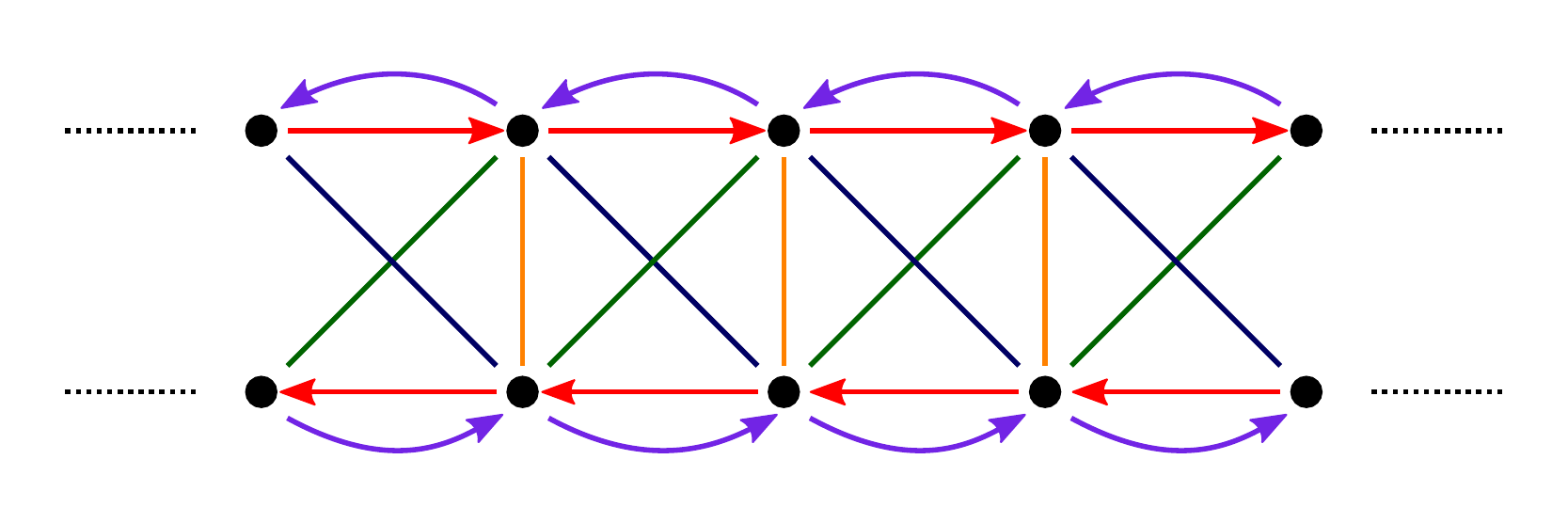}
\caption{(colors online) The most general Cayley graph of the infinite dihedral group which admits a scalar QW with coarse-graining on $\mathbb Z$ with coordination number two. The group presentation is:  $D_{\infty} = \langle a,a^{-1},b,c,d|aa^{-1},b^2,c^2,bda^{-1},cda,bab^{-1}a,cac^{-1}a,bca^{-2}\rangle$. The generators---namely $a$ (red), $a^{-1}$ (violet), $b$ (dark blue), $c$ (green) and $d$ (orange)---are associated to edges of the graph, each corresponding to a transition scalar of the walk. Another Cayley graph of $D_{\infty}$ with the same properties can be obtained by dropping $d$ and the relators containing it. Moreover, one can include $e$ in the generating set, which would correspond to a loop at each site.
}\label{fig:fig2} 
\end{figure}
\begin{figure}
% \includegraphics[width=1\linewidth]{coarse-grained.pdf}
% \caption{(colors online) Realization of the dihedral QW as a quantum walk on $\mathbb{Z}$ with a two-dimensional coin system. The original lattice sites are grouped into pairs, realizing an additional helicity degree of freedom, and each edge (also here $a^{-1}$ is omitted) is associated to a transition matrix. The cosets are associated to an element of a basis for $\mathbb{C}^2$: we choose the canonical basis, that is $c_1\rightarrow (1,0)$ (full circles) and $c_2\rightarrow (0,1)$ (empty circles).
\includegraphics[width=1\linewidth]{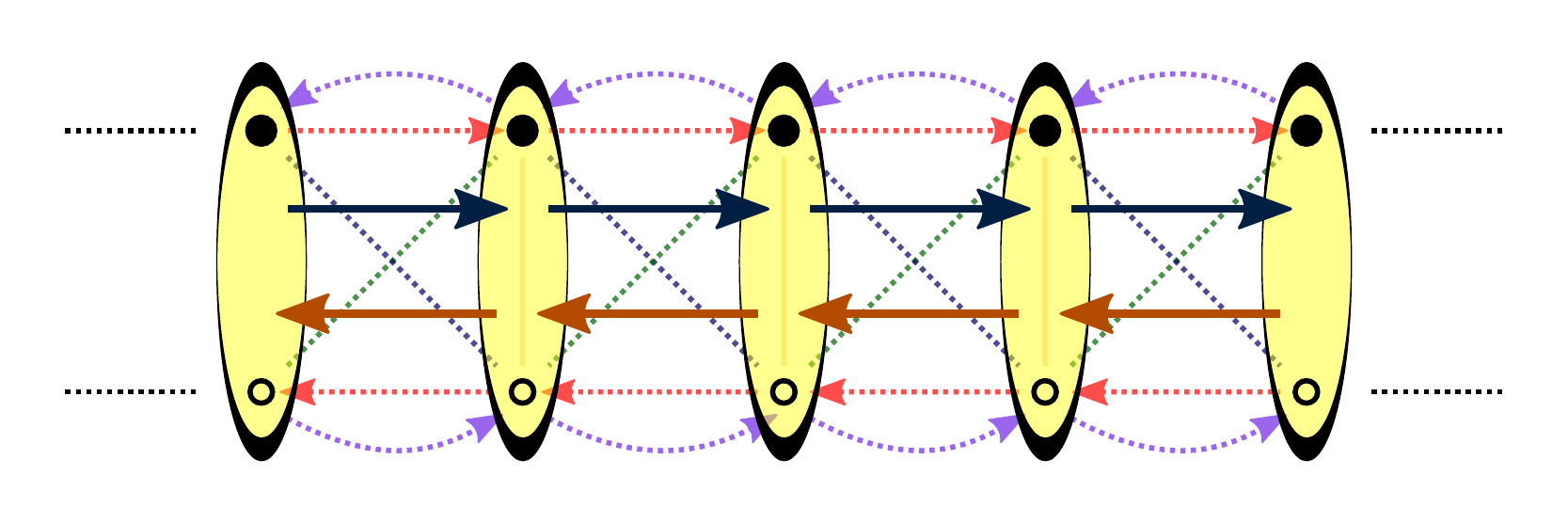}
\caption{(colors online) Realization of the scalar QW in
  Fig. \ref{fig:fig2} as a quantum walk on
  $\mathbb{Z}=\<a,a^{-1}|aa^{-1}\>$ with a two-dimensional coin
  system. The original vertexes are grouped into pairs, realizing an
  additional helicity degree of freedom, and each edge---namely $a$
  (dark blue) and $a^{-1}$ (brown)---is associated to a transition
  matrix. The cosets are associated to an element of a basis for
  $\mathbb{C}^2$: we choose the canonical basis, that is
  $c_1\rightarrow (1,0)$ (full circles) and $c_2\rightarrow (0,1)$
  (empty circles).  }\label{fig:fig3}
\end{figure}

\subsection{Classification of the scalar QWs}\label{dihedral-qw}
The scalar QWs on $D_{\infty}$ are derived in Appendix \ref{der}. The
transition matrices of the coarse-grained QWs, computed choosing $\{ \ket{1},\ket{2} \}$ to be the canonical basis of $\mathbb{C}^2$ and using
equations \eqref{coarse-matr}, are:
\begin{gather*}
A_{+a}=
\begin{pmatrix}
  z_{a}  &  z_{b}  \\
  z_{c}  &  z_{a^{-1}}
\end{pmatrix},\ 
A_{-a}=
\begin{pmatrix}
  z_{a^{-1}} &  z_{c}  \\
  z_{b}  &  z_{a}
\end{pmatrix},\ 
A_e=
\begin{pmatrix}
  z_{e} &  z_{d}  \\
  z_{d}  &  z_{e}
\end{pmatrix}.
\end{gather*}
In Fig. \ref{fig:fig3} one finds a graphical scheme of the
coarse-graining. Given $A_k\coloneqq e^{-ik}A_{+a}+e^{ik}A_{-a}+A_e$,
since the $z_h$ are defined up to an overall phase factor, one can
always take $A_k \in \mathsf{SU}(2)$. 

% Correspondingly, the dispersion relations of $A_k$ read

% \[
% \omega_{\pm}(k)=\pm \arccos\left\lbrace \left(z_a+z_{a^{-1}}\right)\cos k+z_e \right\rbrace .
% \]

The solutions of the unitarity conditions (see Appendix \ref{der})
give that the coarse-grained scalar QWs are of the form
\begin{align}\label{eq:scalar-cg}
A_k =  e^{i\theta\sigma_x} A_k^D e^{i\theta'\sigma_x},
\end{align}
with $\theta,\theta'\in(-\pi/2,0)\cup(0,\pi/2)$, and $A_k^D$ is the Dirac QW in one space dimension \cite{BDT15}
\begin{equation}\label{eq:dirac}
\begin{aligned}
A_k^D=\begin{pmatrix}
   \nu e^{-ik} &  i s\mu  \\
  i s \mu  &  \nu e^{ik}
\end{pmatrix},\quad \nu^2+\mu^2=1,\quad s=\pm 1.
\end{aligned}
\end{equation}
We denote by $\mathcal{W}_{CG}$ the set of coarse-grained scalar QWs,
i.e.  the QWs $A'_k$ such that $A'_k=UA_kU^\dag$ with $A_k$ obeying
Eq.~\eqref{eq:scalar-cg} and $U$ being a local change of basis---say
$U$ does not depend on $k$.

Let us now consider parity invariant QWs $A_{k}^P$ on
$\ell^2(\mathbb{Z})\otimes \mathbb{C}^2$, i.e.
\begin{equation*}
\begin{aligned}
P A_{k}^P P^\dag=A_{-k},\qquad P=P^\dag=P^{-1},
\end{aligned}
\end{equation*}
where $P$ gives a unitary representation in $\mathbb{C}^2$ of the
parity transformation. Assuming that the parity is not represented
trivially, namely $P\neq I$, following the technique of Ref.~\cite{DP14}
one obtains the full characterization of the class of parity invariant QWs:
\begin{equation}\label{eq:parity-qw}
\begin{aligned}
A_k^P=U A'_k U^\dag,\qquad A'_k=e^{i\varphi\sigma_x} A_k^D,
\end{aligned}
\end{equation}
with $\varphi\in[-\pi/2,\pi/2]$, and $U$ a local change of basis. It is immediate to observe that the
walk $e^{i\varphi\sigma_x} A_k^D$ is parity
invariant with $P=\sigma_x$.

We denote by $\mathcal{W}_{P}$ the set of parity invariant QWs, and we
observe that this set coincided with $\mathcal{W}_{CG}$. We have then
proved the following result.

\begin{prp} The set of coarse-grained scalar QWs with coordination number two on $D_\infty$
  coincides with the set of parity invariant QWs on $\ell^2(\mathbb{Z})\otimes \mathbb{C}^2$.
\end{prp}

We notice that the parity symmetry is inherited, in the
coarse-graining procedure, from the particular automorphism $\varphi$ realizing the semidirect product $\mathbb{Z} \rtimes_{\varphi} \mathbb{Z}_2$,
which in this case is the inverse map. The popular Hadamard walk
\cite{AB01}, which is not parity invariant, cannot be obtained by the
coarse-grained of scalar QW.

The dispersion relations of the parity invariant QWs are of the form $\pm \omega(k)$, with
 \begin{align}\label{disp-rel}
   \begin{aligned}
    & \omega(k) = \arccos \left( \delta \cos k + \gamma \right) , \\
     &\delta,\gamma \in \mathbb{R}, |\delta \pm \gamma| \leq1.
   \end{aligned}
\end{align}
\begin{figure}
\includegraphics[width=.8\linewidth]{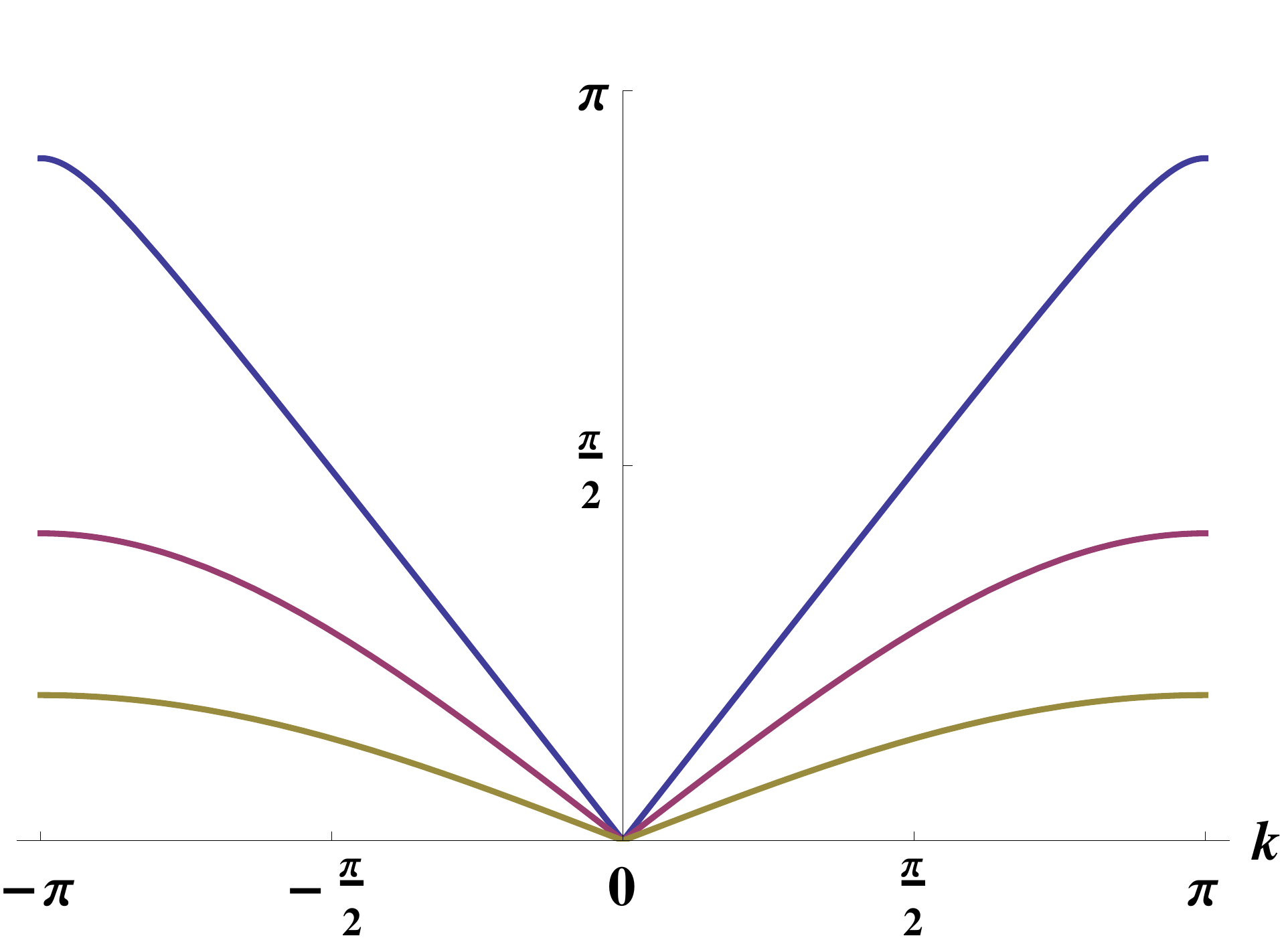}
\caption{(colors online, the quantities plotted are dimensionless) Plot of $\omega(k) = \arccos \left( \delta \cos k + \gamma
  \right)$, for (from top to bottom) $\delta =0.98,0.36,0.09$ and $\delta + \gamma =1$. This class of QWs exhibits,
  for $|k| \approx 0$ and positive $\delta$, a massless (Weyl) dispersion relation (up to a rescaling of $k$): $\omega(k) \approx \delta |k|$. For $|k|\approx\pi$, $\omega(k)$ becomes dispersive (massive).}\label{fig:fig4}
\end{figure}
\begin{figure}
\includegraphics[width=.8\linewidth]{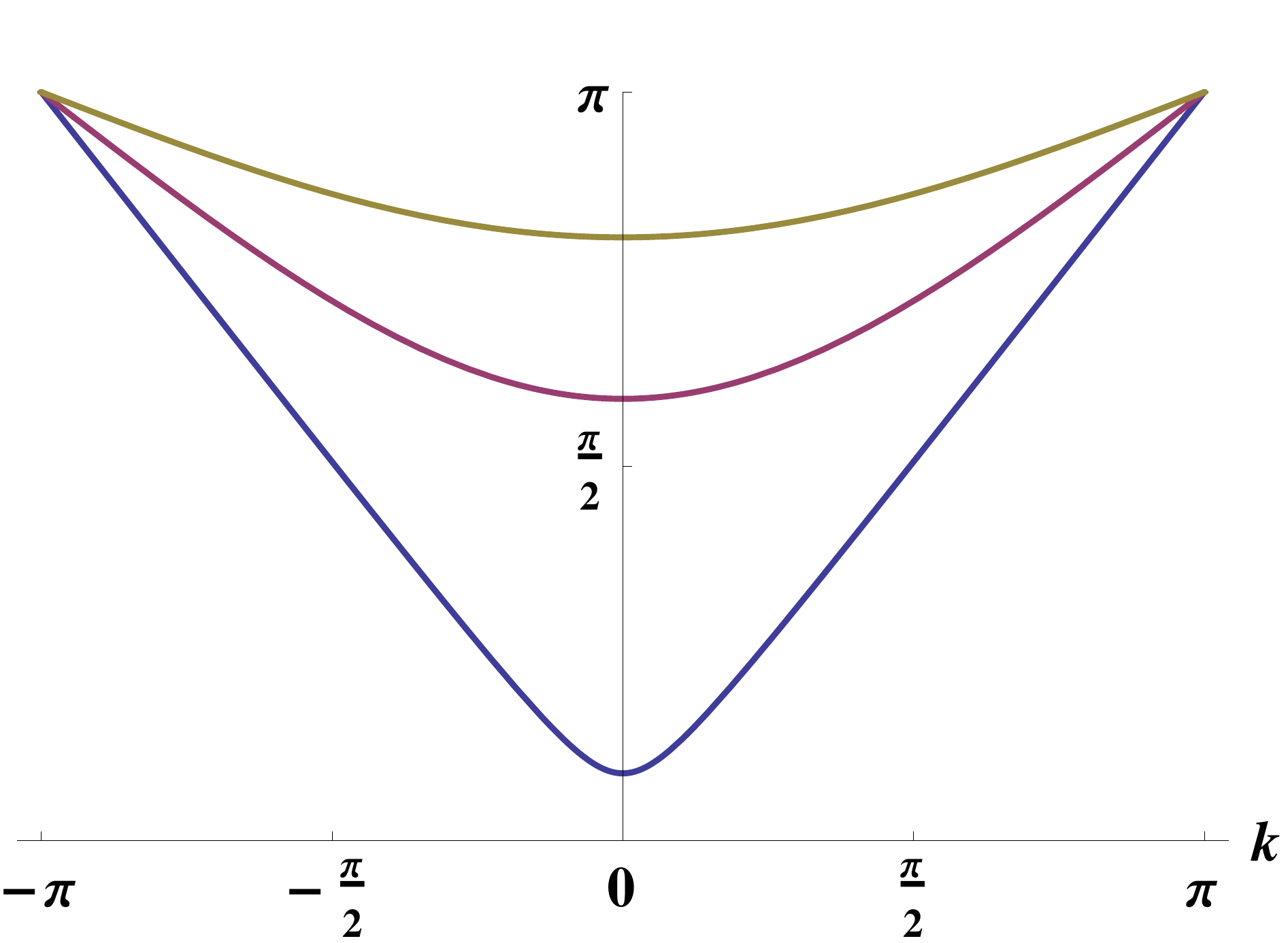}
\caption{(colors online, the quantities plotted are dimensionless)
  Plot of $\omega(k) = \arccos \left( \delta \cos k + \gamma \right)$,
  for (from bottom to top) $\delta =0.98,0.36,0.09$ and
  $\delta - \gamma =1$. This class of QWs exhibits a massive
  dispersion relation for $|k| \approx 0$, and a a massless one for
  $|k| \approx \pi$. The Dirac dispersion relation is recovered for
  $|k|\approx 0$ and $\delta \approx 1$. Notice that these dispersion
  relations are the same as those in Fig. \ref{fig:fig4}, apart from a
  transformation
  $\omega(k)\rightarrow \pi -\omega(k+\pi )$.}\label{fig:fig5}
\end{figure}
For any value of $\gamma,\delta$, the minimum of $\{\omega(k),-\omega(k)\}$ is always attained at
$k_0=0$, and that around $k_0$ the behaviour can be either flat, or $\pm|k|$ plus a constant, or
smooth.  We notice that when $\gamma =0$ we recover (up to a local change of basis) the Dirac QW
$A^D_k$.  When $\delta =1 $ we recover the Weyl QW $A_k = \exp(-ik \sigma_z)$ which describes the
dynamics of massless particles with a dispersion relation which is linear in $k$.  We notice that
when $\delta + \gamma = 1$ the QW exhibits a non-dispersive behaviour for $|k| \approx 0$ and a
dispersive behaviour for greater values of $|k|$.  The dispersion
relations \eqref{disp-rel} are
plotted for some values of the parameters $\delta, \gamma$ in Figs. \ref{fig:fig4},\ref{fig:fig5}.

% Notice that, if $e \not\in S$ ($\gamma = 0$), then from Appendix \ref{der} one has $s = -s'$ and $p=q$, namely \eqref{eq:scalar-cg} reads
% \[
%  A_k = V^p_{s} A_k^D (V^p_{s})^{\dagger},
%  \]
%  and the Dirac QW is obtained up to a local change of basis. Indeed, the $k$-independent term in
%  the dispersion relations vanishes. Referring to Appendix \ref{der}, now $\delta =\sqrt{1-\mu^2}$ with $\mu\in (0,1)$, giving
%  \[
%  \omega(k) = \arccos \left( \sqrt{1-\mu^2} \cos k \right),
%  \]
%  which in the limit of small wave-vector $k$ and small mass parameter $\mu$ approaches the relativistic
%  dispersion relation $\sqrt{k^2 + \mu^2}$. Being $z_e = 0$, the 
%  identity transition in the coarse-grained QW depends merely on the link corresponding to $d$:  this
%  generator realizes the transition from one coset of the scalar QW to the other. In fact $A_e
%  \equiv z_d\sigma_x$ acts as a coupling of two independent scalar Weyl walks, as found in the
%  spin-$1/2$ Abelian case \cite{DP14}. 

% The spinorial Weyl QW, on the other hand, is recovered if $e,d\not\in
% S$ ($\delta=1$), namely from Appendix \ref{der} for $s = -s'$ and
% $p=q=\frac{1}{2}$. This walk exhibits a massless dispersion relation,
% linear in the wave-vector: $\omega(k) = k$. The trivial flat case
% arises for $d \not\in S$ ($\delta =0$): the dispersion relations
% become $k$-independent, giving a trivial dynamics.

\section{Conclusions}
We reviewed the notion of \emph{quantum walk on Cayley graph} with the
focus on scalar QWs. We also reviewed a coarse-graining technique that
allows to unitarily map a scalar QW on a virtually Abelian group to a
coined QW on an Abelian group, what we call a coarse-grained scalar
QW.

The first result we found is a classification of infinite Abelian
scalar QWs (on Cayley graphs with arbitrary presentations), which
turn out to be trivial from a dynamical point of view, meaning that
they are given by a finite direct sum of shift operators times a phase
factor. In particular, this implies that this class of QWs does not
exhibit a massive dispersion relation. This result extends the
previous results of Ref. \cite{MD96}, concerning $\mathbb{Z}^d$.

We then investigated scalar QWs on a the infinite dihedral group
$D_\infty$, which is the unique non-Abelian group $G$ that contains a
subgroup $H\cong \mathbb{Z}$ with index two. We first derived all the
Cayley graphs of $D_\infty$ that allow for a scalar QW. Then we
classified all the admissible QWs over the above Cayley graphs.  

We notice that the unitarity conditions for a QW involve relations
among four generators: if one solves them for the dihedral group, the
derived transition amplitudes are a solution also for the QWs defined
over the dihedral groups $\mathbb{Z}_n \rtimes_{\varphi} \mathbb{Z}_2$
$\forall n \geq 4$, corresponding to the same presentation with the
additional condition $a^n = e$.  In general, transforming some non-cyclic elements
into cyclic ones on infinite groups
presentations allows one to recover QWs on finite groups starting from
QWs on infinite ones.

Finally we have shown that the class of QWs corresponding to the
coarse-grained scalar QWs over $D_\infty$ coincides, up to a local
change of basis, to the class of QWs over
$\ell^2(\mathbb{Z})\otimes \mathbb{C}^2$ that are invariant under
parity transformation. In this class we find QWs whose dispersion
relations can exhibit a massive behaviour. 

% Applying the coarse-graining procedure of Ref. \cite{DEPT15}, we found
% nontrivial massive dynamics.  We were able to connect the dihedral
% scalar walks to some spinorial walks studied in the literature
% \cite{DP14}, namely the Weyl and Dirac QWs. We then also found new
% quantum walks; yet, it is remarkable that all these QWs have a
% physical relativistic limit in the regime of small wave-vectors.

Interestingly, the coarse-graining technique allows one to build a
bridge relating scalar and spinorial quantum walks, studying the
symmetries of the latter as inherited from the underlying Cayley
graph. Existence conditions in scalar QWs are more selective than the
spinorial ones, and one can't recover all possible spinorial QWs
starting from a scalar one. For example the Hadamard QW, which is not
parity invariant, cannot be obtained a coarse-graining of a scalar QW.
On the speculative side, this shows the crucial role played by the coarse-graining in the emergence
of parity symmetry and helicity in one dimension.

\appendix

\acknowledgments
This work has been supported in part by the Templeton Foundation under the project ID\# 43796 \emph{A Quantum-Digital Universe}.

\section{Derivation of the dihedral QWs}\label{der}
Considering the polar representation $z_h = |z_h|e^{i\theta_h}$ for the transition scalars, from the unitarity conditions \eqref{unit} one has (possibly with vanishing $z_d$ or $z_e$)
\begin{gather}
z_{a^{\pm 1}}z_{a^{\mp 1}}^* + z_{g}z_{f}^* = 0, \label{fir'}\\
z_{a^{\pm 1}}z_{g}^* + z_{g}z_{a^{\pm 1}}^* = 0, \label{sec'}\\
z_{d}z_{e}^* + z_{e}z_{d}^* + z_{a^{\pm 1}}z_{b}^* + z_{b}z_{a^{\pm 1}}^* + z_{a^{\mp 1}}z_{c}^* + z_{c}z_{a^{\mp 1}}^* = 0, \label{fou'}\\
z_{g}z_{e}^* + z_{e}z_{g}^* + z_{d}z_{a^{\pm 1}}^* + z_{a^{\pm 1}}z_{d}^* = 0, \label{six} \\
z_{a^{\pm 1}}z_{e}^* + z_{e}z_{a^{\mp 1}}^* + z_{d}z_{g}^* + z_{f}z_{d}^* = 0, \label{eig}
\end{gather}
with $g,f\in \{b,c\}$ and $g\neq f$. From eqs. \eqref{sec'}, it follows \emph{e.g.} $e^{i\theta_b}=t_1ie^{i\theta_{a^{-1}}}$ and $e^{i\theta_c}=t_2ie^{i\theta_{a}}$ ($t_{1,2}$ arbitrary signs), while from \eqref{fir'} one has:
\begin{gather*}
|z_{a}||z_{a^{-1}}|=|z_b||z_c|, \\
s_1\coloneqq t_1=-t_2, \\
e^{i\theta} \coloneqq  e^{i\theta_{a}} = s_2e^{i\theta_{a^{-1}}},
\end{gather*}
which are consistent with all of the \eqref{sec'}.

Therefore, we can satisfy the previous conditions and the normalization in \eqref{unit} arbitrarily defining some real parameters such that the transition scalars are given by:
\begin{align*}
z_{a}&=\nu\sqrt{p}\sqrt{q}e^{i\theta},\ &z_{a^{-1}}&=s_2\nu\sqrt{1-p}\sqrt{1-q}e^{i\theta}, \\
z_{b}&=s_2s_1i\nu\sqrt{p}\sqrt{1-q}e^{i\theta},\ &z_{c}&=-s_1i\nu\sqrt{1-p}\sqrt{q}e^{i\theta}, \\
z_{e} & = \mu \alpha e^{i\theta_e},\ & z_{d} & = \mu\beta e^{i\theta_d},
\end{align*}
where $p,q\in (0,1)$, $\mu \in [0,1)$, $\alpha \in [0,1]$ and $\nu\coloneqq \sqrt{1-\mu^2}, \beta \coloneqq \sqrt{1-\alpha^2}$.
\paragraph{{\bf Case} $z_d=z_e=0$ ($\mu = 0$).} The \eqref{fou'} are already satisfied, while eqs. \eqref{six},\eqref{eig} turn out to be trivial. Finally, we conclude that the transition scalars are:
\begin{gather*}
z_{a}=\sqrt{p}\sqrt{q}e^{i\theta},\ z_{a^{-1}}=s_2\sqrt{1-p}\sqrt{1-q}e^{i\theta}, \\
z_{b}=s_2s_1i\sqrt{p}\sqrt{1-q}e^{i\theta},\ z_{c}=-s_1i\sqrt{1-p}\sqrt{q}e^{i\theta}.
\end{gather*}
\paragraph{{\bf Case} $z_e=0$ ($\alpha = 0$).} The \eqref{fou'} are already satisfied. From \eqref{six} one has $e^{i\theta_d}=s_3ie^{i\theta}$, while from \eqref{eig}
\begin{gather*}
e^{2i\theta_d} = -s_2e^{2i\theta}\ \Rightarrow\ s_2=+1, \\
|z_b|=|z_c|.
\end{gather*}
We conclude that the transition scalars are (up to a global phase factor)
\begin{gather*}
z_{a}=\nu p,\ z_{a^{-1}}=\nu (1-p), \\
z_{b}=s_1i\nu\sqrt{p}\sqrt{1-p},\ z_{c}=-s_1i\nu\sqrt{p}\sqrt{1-p},\\ z_d = s_3i\mu ,
\end{gather*}
for $p,\mu \in (0,1)$.
\paragraph{{\bf Case} $z_d=0$ ($\beta = 0$).} The \eqref{fou'} are already satisfied. From \eqref{six} one has $e^{i\theta_e}=s'_3e^{i\theta}\coloneqq -s_1s_3e^{i\theta}$ (this definition is convenient in view of the next case $z_e,z_d\neq 0$). From \eqref{eig} one gets
\begin{gather*}
e^{2i\theta_e}=-s_2e^{2i\theta}\ \Rightarrow\ s_2=-1, \\
|z_{a}|=|z_{a^{-1}}| .
\end{gather*}
We thus conclude that the transition scalars are (up to a global phase factor)
\begin{gather*}
z_{a} =\nu\sqrt{p}\sqrt{1-p},\ z_{a^{-1}} =-\nu\sqrt{p}\sqrt{1-p}, \\
z_{b} =-s_1i\nu p,\ z_{c} =-s_1i\nu (1-p),\ z_e = -s_1s_3\mu
\end{gather*}
for $p,\mu\in (0,1)$.
\paragraph{{\bf Case} $z_e,z_d\neq 0$ ($\mu,\alpha,\beta \neq 0$).} \eqref{fou'} reads
\[
z_{d}z_{e}^* + z_{e}z_{d}^* = 0 \ \Rightarrow\ e^{i\theta_e} = s_4ie^{i\theta_d}.
\]
Substituting in eqs. \eqref{six}, one has
\begin{gather*}
s_1s_2s_4|z_b||z_e|\cos(\theta_d-\theta)=-|z_d||z_a|\cos(\theta_d-\theta), \\
s_1s_2s_4|z_c||z_e|\cos(\theta_d-\theta)=|z_d||z_{a^{-1}}|\cos(\theta_d-\theta),
\end{gather*}
which can be satisfied only if $\cos(\theta_d-\theta) =0$, implying
that $e^{i\theta_d}=s_3ie^{i\theta}$. From \eqref{eig} we have
\begin{equation}\label{g}
s_1|z_e|(|z_a|+s_2|z_{a^{-1}}|)=s_4|z_d|(s_2|z_b|-|z_c|).
\end{equation}
Notice that a change of the sign $s_1s_4$ affects last equation just by a relabeling $|z_a| \leftrightarrow |z_ {a^{-1}}|$ (if $s_2=-1$) or $|z_b| \leftrightarrow |z_c|$ (if $s_2=+1$), under which the unitarity conditions are invariant: thus we can set $s_1s_4=+1$. From \eqref{g}, one has to impose a positivity condition according to the choice of $s_2$, i.e. $|z_b|-|z_c| > 0$ or $|z_{a^{-1}}| - |z_a|>0$, which imply some conditions on the $p,q$. Finally, one also finds the expression of $\alpha $ in terms of $p,q$ and $s_2$. We conclude that the transition scalars are (up to a phase factor):
\begin{align*}
z_{a} & =\nu\sqrt{p}\sqrt{q},\ & z_{a^{-1}} & =s_2\nu\sqrt{1-p}\sqrt{1-q}, \\
z_{b} & =s_2s_1i\nu\sqrt{p}\sqrt{1-q},\ & z_{c} & =-s_1i\nu\sqrt{1-p}\sqrt{q}, \\
z_{e} & =-s_1s_3\mu \alpha ,\ & z_{d} & =s_3i\mu \beta ,
\end{align*}
where $\alpha =\sqrt{p}\sqrt{1-q}-s_2\sqrt{1-p}\sqrt{q}$ and $p,q,\mu\in (0,1)$, while $p>q$ if $s_2=+1$, $(1-q)>p$ if $s_2=-1$.

In Eq. \eqref{eq:scalar-cg} of Sec. \ref{dihedral-qw} we posed $\cos\theta:=\sqrt{p}$, $\sin\theta:=-s_1\sqrt{1-p}$, $\cos\theta':=\sqrt{q}$, $\sin\theta'=s_1s_2\sqrt{1-q}$, in Eq. \eqref{eq:dirac} $s\coloneqq s_2s_3$, and in Eq. \eqref{disp-rel} $\delta \coloneqq z_a+z_{a^{-1}}$ and $\gamma \coloneqq z_e$.

%in general:
%\begin{gather*}
%\text{if}\ s_2=+1\ \text{then} \\
% |z_b|-|z_c| > 0\ \Rightarrow\ p > q; \\
%\text{if}\ s_2=-1\ \text{then} \\
%|z_{a^{-1}}| - |z_a| > 0\ \Rightarrow\ (1-q) > p .
%\end{gather*}

\bibliography{scalar-qws}

\begin{thebibliography}{32}%
\makeatletter
\providecommand \@ifxundefined [1]{%
 \@ifx{#1\undefined}
}%
\providecommand \@ifnum [1]{%
 \ifnum #1\expandafter \@firstoftwo
 \else \expandafter \@secondoftwo
 \fi
}%
\providecommand \@ifx [1]{%
 \ifx #1\expandafter \@firstoftwo
 \else \expandafter \@secondoftwo
 \fi
}%
\providecommand \natexlab [1]{#1}%
\providecommand \enquote  [1]{``#1''}%
\providecommand \bibnamefont  [1]{#1}%
\providecommand \bibfnamefont [1]{#1}%
\providecommand \citenamefont [1]{#1}%
\providecommand \href@noop [0]{\@secondoftwo}%
\providecommand \href [0]{\begingroup \@sanitize@url \@href}%
\providecommand \@href[1]{\@@startlink{#1}\@@href}%
\providecommand \@@href[1]{\endgroup#1\@@endlink}%
\providecommand \@sanitize@url [0]{\catcode `\\12\catcode `\$12\catcode
  `\&12\catcode `\#12\catcode `\^12\catcode `\_12\catcode `\%12\relax}%
\providecommand \@@startlink[1]{}%
\providecommand \@@endlink[0]{}%
\providecommand \url  [0]{\begingroup\@sanitize@url \@url }%
\providecommand \@url [1]{\endgroup\@href {#1}{\urlprefix }}%
\providecommand \urlprefix  [0]{URL }%
\providecommand \Eprint [0]{\href }%
\providecommand \doibase [0]{http://dx.doi.org/}%
\providecommand \selectlanguage [0]{\@gobble}%
\providecommand \bibinfo  [0]{\@secondoftwo}%
\providecommand \bibfield  [0]{\@secondoftwo}%
\providecommand \translation [1]{[#1]}%
\providecommand \BibitemOpen [0]{}%
\providecommand \bibitemStop [0]{}%
\providecommand \bibitemNoStop [0]{.\EOS\space}%
\providecommand \EOS [0]{\spacefactor3000\relax}%
\providecommand \BibitemShut  [1]{\csname bibitem#1\endcsname}%
\let\auto@bib@innerbib\@empty
%</preamble>
\bibitem [{\citenamefont {Einstein}(1905)}]{E05}%
  \BibitemOpen
  \bibfield  {author} {\bibinfo {author} {\bibfnamefont {A.}~\bibnamefont
  {Einstein}},\ }\href@noop {} {\bibfield  {journal} {\bibinfo  {journal}
  {Annalen der Physik}\ }\textbf {\bibinfo {volume} {17}},\ \bibinfo {pages}
  {549} (\bibinfo {year} {1905})}\BibitemShut {NoStop}%
\bibitem [{\citenamefont {Aharonov}\ \emph {et~al.}(1993)\citenamefont
  {Aharonov}, \citenamefont {Davidovich},\ and\ \citenamefont
  {Zagury}}]{ADZ93}%
  \BibitemOpen
  \bibfield  {author} {\bibinfo {author} {\bibfnamefont {Y.}~\bibnamefont
  {Aharonov}}, \bibinfo {author} {\bibfnamefont {L.}~\bibnamefont
  {Davidovich}}, \ and\ \bibinfo {author} {\bibfnamefont {N.}~\bibnamefont
  {Zagury}},\ }\href {\doibase 10.1103/PhysRevA.48.1687} {\bibfield  {journal}
  {\bibinfo  {journal} {Phys. Rev. A}\ }\textbf {\bibinfo {volume} {48}},\
  \bibinfo {pages} {1687} (\bibinfo {year} {1993})}\BibitemShut {NoStop}%
\bibitem [{\citenamefont {Meyer}(1996{\natexlab{a}})}]{M96}%
  \BibitemOpen
  \bibfield  {author} {\bibinfo {author} {\bibfnamefont {D.~A.}\ \bibnamefont
  {Meyer}},\ }\href {\doibase 10.1007/BF02199356} {\bibfield  {journal}
  {\bibinfo  {journal} {Journal of Statistical Physics}\ }\textbf {\bibinfo
  {volume} {85}},\ \bibinfo {pages} {551} (\bibinfo {year}
  {1996}{\natexlab{a}})}\BibitemShut {NoStop}%
\bibitem [{\citenamefont {Ambainis}\ \emph {et~al.}(2001)\citenamefont
  {Ambainis}, \citenamefont {Bach}, \citenamefont {Nayak}, \citenamefont
  {Vishwanath},\ and\ \citenamefont {Watrous}}]{AB01}%
  \BibitemOpen
  \bibfield  {author} {\bibinfo {author} {\bibfnamefont {A.}~\bibnamefont
  {Ambainis}}, \bibinfo {author} {\bibfnamefont {E.}~\bibnamefont {Bach}},
  \bibinfo {author} {\bibfnamefont {A.}~\bibnamefont {Nayak}}, \bibinfo
  {author} {\bibfnamefont {A.}~\bibnamefont {Vishwanath}}, \ and\ \bibinfo
  {author} {\bibfnamefont {J.}~\bibnamefont {Watrous}},\ }in\ \href {\doibase
  10.1145/380752.380757} {\emph {\bibinfo {booktitle} {Proceedings of the
  Thirty-third Annual ACM Symposium on Theory of Computing}}},\ \bibinfo
  {series and number} {STOC '01}\ (\bibinfo  {publisher} {ACM},\ \bibinfo
  {address} {New York, NY, USA},\ \bibinfo {year} {2001})\ pp.\ \bibinfo
  {pages} {37--49}\BibitemShut {NoStop}%
\bibitem [{\citenamefont {Severini}(2003)}]{S03}%
  \BibitemOpen
  \bibfield  {author} {\bibinfo {author} {\bibfnamefont {S.}~\bibnamefont
  {Severini}},\ }\href {\doibase 10.1137/S0895479802410293} {\bibfield
  {journal} {\bibinfo  {journal} {SIAM Journal on Matrix Analysis and
  Applications}\ }\textbf {\bibinfo {volume} {25}},\ \bibinfo {pages} {295}
  (\bibinfo {year} {2003})}\BibitemShut {NoStop}%
\bibitem [{\citenamefont {Knight}\ \emph {et~al.}(2004)\citenamefont {Knight},
  \citenamefont {Rold\'an},\ and\ \citenamefont {Sipe}}]{KRS04}%
  \BibitemOpen
  \bibfield  {author} {\bibinfo {author} {\bibfnamefont {P.~L.}\ \bibnamefont
  {Knight}}, \bibinfo {author} {\bibfnamefont {E.}~\bibnamefont {Rold\'an}}, \
  and\ \bibinfo {author} {\bibfnamefont {J.~E.}\ \bibnamefont {Sipe}},\ }\href
  {\doibase 10.1080/09500340408232489} {\bibfield  {journal} {\bibinfo
  {journal} {Journal of Modern Optics}\ }\textbf {\bibinfo {volume} {51}},\
  \bibinfo {pages} {1761} (\bibinfo {year} {2004})}\BibitemShut {NoStop}%
\bibitem [{\citenamefont {Montanaro}(2007)}]{M07}%
  \BibitemOpen
  \bibfield  {author} {\bibinfo {author} {\bibfnamefont {A.}~\bibnamefont
  {Montanaro}},\ }\href {http://dl.acm.org/citation.cfm?id=2011706.2011711}
  {\bibfield  {journal} {\bibinfo  {journal} {Quantum Info. Comput.}\ }\textbf
  {\bibinfo {volume} {7}},\ \bibinfo {pages} {93} (\bibinfo {year}
  {2007})}\BibitemShut {NoStop}%
\bibitem [{\citenamefont {Kogut}(1983)}]{K83}%
  \BibitemOpen
  \bibfield  {author} {\bibinfo {author} {\bibfnamefont {J.~B.}\ \bibnamefont
  {Kogut}},\ }\href {\doibase 10.1103/RevModPhys.55.775} {\bibfield  {journal}
  {\bibinfo  {journal} {Rev. Mod. Phys.}\ }\textbf {\bibinfo {volume} {55}},\
  \bibinfo {pages} {775} (\bibinfo {year} {1983})}\BibitemShut {NoStop}%
\bibitem [{\citenamefont {Chin}\ \emph {et~al.}(1984)\citenamefont {Chin},
  \citenamefont {Negele},\ and\ \citenamefont {Koonin}}]{CNK84}%
  \BibitemOpen
  \bibfield  {author} {\bibinfo {author} {\bibfnamefont {S.}~\bibnamefont
  {Chin}}, \bibinfo {author} {\bibfnamefont {J.}~\bibnamefont {Negele}}, \ and\
  \bibinfo {author} {\bibfnamefont {S.}~\bibnamefont {Koonin}},\ }\href
  {\doibase http://dx.doi.org/10.1016/0003-4916(84)90050-2} {\bibfield
  {journal} {\bibinfo  {journal} {Annals of Physics}\ }\textbf {\bibinfo
  {volume} {157}},\ \bibinfo {pages} {140 } (\bibinfo {year}
  {1984})}\BibitemShut {NoStop}%
\bibitem [{\citenamefont {Arnault}\ and\ \citenamefont
  {Debbasch}(2016)}]{AD16}%
  \BibitemOpen
  \bibfield  {author} {\bibinfo {author} {\bibfnamefont {P.}~\bibnamefont
  {Arnault}}\ and\ \bibinfo {author} {\bibfnamefont {F.}~\bibnamefont
  {Debbasch}},\ }\href@noop {} {\bibfield  {journal} {\bibinfo  {journal}
  {\href{http://arxiv.org/abs/1508.00038v4}{arXiv:1508.00038v4 [quant-ph]}}\ }
  (\bibinfo {year} {2016})}\BibitemShut {NoStop}%
\bibitem [{\citenamefont {Childs}\ \emph {et~al.}(2003)\citenamefont {Childs},
  \citenamefont {Cleve}, \citenamefont {Deotto}, \citenamefont {Farhi},
  \citenamefont {Gutmann},\ and\ \citenamefont {Spielman}}]{CCD03}%
  \BibitemOpen
  \bibfield  {author} {\bibinfo {author} {\bibfnamefont {A.~M.}\ \bibnamefont
  {Childs}}, \bibinfo {author} {\bibfnamefont {R.}~\bibnamefont {Cleve}},
  \bibinfo {author} {\bibfnamefont {E.}~\bibnamefont {Deotto}}, \bibinfo
  {author} {\bibfnamefont {E.}~\bibnamefont {Farhi}}, \bibinfo {author}
  {\bibfnamefont {S.}~\bibnamefont {Gutmann}}, \ and\ \bibinfo {author}
  {\bibfnamefont {D.~A.}\ \bibnamefont {Spielman}},\ }in\ \href {\doibase
  10.1145/780542.780552} {\emph {\bibinfo {booktitle} {Proceedings of the
  Thirty-fifth Annual ACM Symposium on Theory of Computing}}},\ \bibinfo
  {series and number} {STOC '03}\ (\bibinfo  {publisher} {ACM},\ \bibinfo
  {address} {New York, NY, USA},\ \bibinfo {year} {2003})\ pp.\ \bibinfo
  {pages} {59--68}\BibitemShut {NoStop}%
\bibitem [{\citenamefont {Ambainis}(2004)}]{A04}%
  \BibitemOpen
  \bibfield  {author} {\bibinfo {author} {\bibfnamefont {A.}~\bibnamefont
  {Ambainis}},\ }in\ \href {\doibase 10.1109/FOCS.2004.54} {\emph {\bibinfo
  {booktitle} {Foundations of Computer Science, 2004. Proceedings. 45th Annual
  IEEE Symposium on}}}\ (\bibinfo {year} {2004})\ pp.\ \bibinfo {pages}
  {22--31}\BibitemShut {NoStop}%
\bibitem [{\citenamefont {Magniez}\ \emph {et~al.}(2005)\citenamefont
  {Magniez}, \citenamefont {Santha},\ and\ \citenamefont {Szegedy}}]{MFS05}%
  \BibitemOpen
  \bibfield  {author} {\bibinfo {author} {\bibfnamefont {F.}~\bibnamefont
  {Magniez}}, \bibinfo {author} {\bibfnamefont {M.}~\bibnamefont {Santha}}, \
  and\ \bibinfo {author} {\bibfnamefont {M.}~\bibnamefont {Szegedy}},\ }in\
  \href {http://dl.acm.org/citation.cfm?id=1070432.1070591} {\emph {\bibinfo
  {booktitle} {Proceedings of the Sixteenth Annual ACM-SIAM Symposium on
  Discrete Algorithms}}}\ (\bibinfo {year} {2005})\ pp.\ \bibinfo {pages}
  {1109--1117}\BibitemShut {NoStop}%
\bibitem [{\citenamefont {Bisio}\ \emph
  {et~al.}(2015{\natexlab{a}})\citenamefont {Bisio}, \citenamefont {D'Ariano},\
  and\ \citenamefont {Tosini}}]{BDT15}%
  \BibitemOpen
  \bibfield  {author} {\bibinfo {author} {\bibfnamefont {A.}~\bibnamefont
  {Bisio}}, \bibinfo {author} {\bibfnamefont {G.~M.}\ \bibnamefont {D'Ariano}},
  \ and\ \bibinfo {author} {\bibfnamefont {A.}~\bibnamefont {Tosini}},\ }\href
  {\doibase http://dx.doi.org/10.1016/j.aop.2014.12.016} {\bibfield  {journal}
  {\bibinfo  {journal} {Annals of Physics}\ }\textbf {\bibinfo {volume}
  {354}},\ \bibinfo {pages} {244 } (\bibinfo {year}
  {2015}{\natexlab{a}})}\BibitemShut {NoStop}%
\bibitem [{\citenamefont {D'Ariano}\ and\ \citenamefont
  {Perinotti}(2014)}]{DP14}%
  \BibitemOpen
  \bibfield  {author} {\bibinfo {author} {\bibfnamefont {G.~M.}\ \bibnamefont
  {D'Ariano}}\ and\ \bibinfo {author} {\bibfnamefont {P.}~\bibnamefont
  {Perinotti}},\ }\href {\doibase 10.1103/PhysRevA.90.062106} {\bibfield
  {journal} {\bibinfo  {journal} {Phys. Rev. A}\ }\textbf {\bibinfo {volume}
  {90}},\ \bibinfo {pages} {062106} (\bibinfo {year} {2014})}\BibitemShut
  {NoStop}%
\bibitem [{\citenamefont {Arrighi}\ \emph {et~al.}(2014)\citenamefont
  {Arrighi}, \citenamefont {Facchini},\ and\ \citenamefont
  {Forets}}]{ArrighiNJP14}%
  \BibitemOpen
  \bibfield  {author} {\bibinfo {author} {\bibfnamefont {P.}~\bibnamefont
  {Arrighi}}, \bibinfo {author} {\bibfnamefont {S.}~\bibnamefont {Facchini}}, \
  and\ \bibinfo {author} {\bibfnamefont {M.}~\bibnamefont {Forets}},\ }\href
  {http://stacks.iop.org/1367-2630/16/i=9/a=093007} {\bibfield  {journal}
  {\bibinfo  {journal} {New Journal of Physics}\ }\textbf {\bibinfo {volume}
  {16}},\ \bibinfo {pages} {093007} (\bibinfo {year} {2014})}\BibitemShut
  {NoStop}%
\bibitem [{\citenamefont {Bibeau-Delisle}\ \emph {et~al.}(2015)\citenamefont
  {Bibeau-Delisle}, \citenamefont {Bisio}, \citenamefont {D'Ariano},
  \citenamefont {Perinotti},\ and\ \citenamefont {Tosini}}]{BBDPT15}%
  \BibitemOpen
  \bibfield  {author} {\bibinfo {author} {\bibfnamefont {A.}~\bibnamefont
  {Bibeau-Delisle}}, \bibinfo {author} {\bibfnamefont {A.}~\bibnamefont
  {Bisio}}, \bibinfo {author} {\bibfnamefont {G.~M.}\ \bibnamefont {D'Ariano}},
  \bibinfo {author} {\bibfnamefont {P.}~\bibnamefont {Perinotti}}, \ and\
  \bibinfo {author} {\bibfnamefont {A.}~\bibnamefont {Tosini}},\ }\href
  {http://stacks.iop.org/0295-5075/109/i=5/a=50003} {\bibfield  {journal}
  {\bibinfo  {journal} {EPL (Europhysics Letters)}\ }\textbf {\bibinfo {volume}
  {109}},\ \bibinfo {pages} {50003} (\bibinfo {year} {2015})}\BibitemShut
  {NoStop}%
\bibitem [{\citenamefont {Bisio}\ \emph
  {et~al.}(2015{\natexlab{b}})\citenamefont {Bisio}, \citenamefont {D'Ariano},\
  and\ \citenamefont {Perinotti}}]{BDP15}%
  \BibitemOpen
  \bibfield  {author} {\bibinfo {author} {\bibfnamefont {A.}~\bibnamefont
  {Bisio}}, \bibinfo {author} {\bibfnamefont {G.~M.}\ \bibnamefont {D'Ariano}},
  \ and\ \bibinfo {author} {\bibfnamefont {P.}~\bibnamefont {Perinotti}},\
  }\href@noop {} {\bibfield  {journal} {\bibinfo  {journal}
  {\href{http://arxiv.org/abs/1503.01017v1}{arXiv:1503.01017v1 [quant-ph]}}\ }
  (\bibinfo {year} {2015}{\natexlab{b}})}\BibitemShut {NoStop}%
\bibitem [{\citenamefont {Arrighi}\ \emph {et~al.}(2015)\citenamefont
  {Arrighi}, \citenamefont {Facchini},\ and\ \citenamefont {Forets}}]{AFF15}%
  \BibitemOpen
  \bibfield  {author} {\bibinfo {author} {\bibfnamefont {P.}~\bibnamefont
  {Arrighi}}, \bibinfo {author} {\bibfnamefont {S.}~\bibnamefont {Facchini}}, \
  and\ \bibinfo {author} {\bibfnamefont {M.}~\bibnamefont {Forets}},\
  }\href@noop {} {\bibfield  {journal} {\bibinfo  {journal}
  {\href{http://arxiv.org/abs/1505.07023v1}{arXiv:1505.07023v1 [quant-ph]}}\ }
  (\bibinfo {year} {2015})}\BibitemShut {NoStop}%
\bibitem [{\citenamefont {Bialynicki-Birula}(1994)}]{BB94}%
  \BibitemOpen
  \bibfield  {author} {\bibinfo {author} {\bibfnamefont {I.}~\bibnamefont
  {Bialynicki-Birula}},\ }\href {\doibase 10.1103/PhysRevD.49.6920} {\bibfield
  {journal} {\bibinfo  {journal} {Phys. Rev. D}\ }\textbf {\bibinfo {volume}
  {49}},\ \bibinfo {pages} {6920} (\bibinfo {year} {1994})}\BibitemShut
  {NoStop}%
\bibitem [{\citenamefont {Farrelly}\ and\ \citenamefont {Short}(2014)}]{FS14}%
  \BibitemOpen
  \bibfield  {author} {\bibinfo {author} {\bibfnamefont {T.~C.}\ \bibnamefont
  {Farrelly}}\ and\ \bibinfo {author} {\bibfnamefont {A.~J.}\ \bibnamefont
  {Short}},\ }\href {\doibase 10.1103/PhysRevA.89.062109} {\bibfield  {journal}
  {\bibinfo  {journal} {Phys. Rev. A}\ }\textbf {\bibinfo {volume} {89}},\
  \bibinfo {pages} {062109} (\bibinfo {year} {2014})}\BibitemShut {NoStop}%
\bibitem [{\citenamefont {Bisio}\ \emph {et~al.}(2016)\citenamefont {Bisio},
  \citenamefont {D'Ariano},\ and\ \citenamefont {Perinotti}}]{BDP16}%
  \BibitemOpen
  \bibfield  {author} {\bibinfo {author} {\bibfnamefont {A.}~\bibnamefont
  {Bisio}}, \bibinfo {author} {\bibfnamefont {G.~M.}\ \bibnamefont {D'Ariano}},
  \ and\ \bibinfo {author} {\bibfnamefont {P.}~\bibnamefont {Perinotti}},\
  }\href {\doibase http://dx.doi.org/10.1016/j.aop.2016.02.009} {\bibfield
  {journal} {\bibinfo  {journal} {Annals of Physics}\ }\textbf {\bibinfo
  {volume} {368}},\ \bibinfo {pages} {177 } (\bibinfo {year}
  {2016})}\BibitemShut {NoStop}%
\bibitem [{\citenamefont {D'Ariano}\ \emph {et~al.}(2015)\citenamefont
  {D'Ariano}, \citenamefont {Erba}, \citenamefont {Perinotti},\ and\
  \citenamefont {Tosini}}]{DEPT15}%
  \BibitemOpen
  \bibfield  {author} {\bibinfo {author} {\bibfnamefont {G.~M.}\ \bibnamefont
  {D'Ariano}}, \bibinfo {author} {\bibfnamefont {M.}~\bibnamefont {Erba}},
  \bibinfo {author} {\bibfnamefont {P.}~\bibnamefont {Perinotti}}, \ and\
  \bibinfo {author} {\bibfnamefont {A.}~\bibnamefont {Tosini}},\ }\href@noop {}
  {\bibfield  {journal} {\bibinfo  {journal}
  {\href{http://arxiv.org/abs/1511.03992v1}{arXiv:1511.03992v1 [quant-ph]}}\ }
  (\bibinfo {year} {2015})}\BibitemShut {NoStop}%
\bibitem [{\citenamefont {Patel}\ \emph {et~al.}(2005)\citenamefont {Patel},
  \citenamefont {Raghunathan},\ and\ \citenamefont
  {Rungta}}]{PhysRevA.71.032347}%
  \BibitemOpen
  \bibfield  {author} {\bibinfo {author} {\bibfnamefont {A.}~\bibnamefont
  {Patel}}, \bibinfo {author} {\bibfnamefont {K.~S.}\ \bibnamefont
  {Raghunathan}}, \ and\ \bibinfo {author} {\bibfnamefont {P.}~\bibnamefont
  {Rungta}},\ }\href {\doibase 10.1103/PhysRevA.71.032347} {\bibfield
  {journal} {\bibinfo  {journal} {Phys. Rev. A}\ }\textbf {\bibinfo {volume}
  {71}},\ \bibinfo {pages} {032347} (\bibinfo {year} {2005})}\BibitemShut
  {NoStop}%
\bibitem [{\citenamefont {Portugal}\ \emph {et~al.}(2015)\citenamefont
  {Portugal}, \citenamefont {Boettcher},\ and\ \citenamefont
  {Falkner}}]{PhysRevA.91.052319}%
  \BibitemOpen
  \bibfield  {author} {\bibinfo {author} {\bibfnamefont {R.}~\bibnamefont
  {Portugal}}, \bibinfo {author} {\bibfnamefont {S.}~\bibnamefont {Boettcher}},
  \ and\ \bibinfo {author} {\bibfnamefont {S.}~\bibnamefont {Falkner}},\ }\href
  {\doibase 10.1103/PhysRevA.91.052319} {\bibfield  {journal} {\bibinfo
  {journal} {Phys. Rev. A}\ }\textbf {\bibinfo {volume} {91}},\ \bibinfo
  {pages} {052319} (\bibinfo {year} {2015})}\BibitemShut {NoStop}%
\bibitem [{\citenamefont {Santos}\ \emph {et~al.}(2015)\citenamefont {Santos},
  \citenamefont {Portugal},\ and\ \citenamefont
  {Boettcher}}]{Santos:2015:MCQ:2822142.2822148}%
  \BibitemOpen
  \bibfield  {author} {\bibinfo {author} {\bibfnamefont {R.~A.}\ \bibnamefont
  {Santos}}, \bibinfo {author} {\bibfnamefont {R.}~\bibnamefont {Portugal}}, \
  and\ \bibinfo {author} {\bibfnamefont {S.}~\bibnamefont {Boettcher}},\ }\href
  {\doibase 10.1007/s11128-015-1042-9} {\bibfield  {journal} {\bibinfo
  {journal} {Quantum Information Processing}\ }\textbf {\bibinfo {volume}
  {14}},\ \bibinfo {pages} {3179} (\bibinfo {year} {2015})}\BibitemShut
  {NoStop}%
\bibitem [{\citenamefont {Acevedo}\ \emph {et~al.}(2008)\citenamefont
  {Acevedo}, \citenamefont {Roland},\ and\ \citenamefont {Cerf}}]{ARC08}%
  \BibitemOpen
  \bibfield  {author} {\bibinfo {author} {\bibfnamefont {O.~L.}\ \bibnamefont
  {Acevedo}}, \bibinfo {author} {\bibfnamefont {J.}~\bibnamefont {Roland}}, \
  and\ \bibinfo {author} {\bibfnamefont {N.~J.}\ \bibnamefont {Cerf}},\ }\href
  {http://dl.acm.org/citation.cfm?id=2011752.2011757} {\bibfield  {journal}
  {\bibinfo  {journal} {Quantum Info. Comput.}\ }\textbf {\bibinfo {volume}
  {8}},\ \bibinfo {pages} {68} (\bibinfo {year} {2008})}\BibitemShut {NoStop}%
\bibitem [{\citenamefont {Meyer}(1996{\natexlab{b}})}]{MD96}%
  \BibitemOpen
  \bibfield  {author} {\bibinfo {author} {\bibfnamefont {D.~A.}\ \bibnamefont
  {Meyer}},\ }\href {\doibase http://dx.doi.org/10.1016/S0375-9601(96)00745-1}
  {\bibfield  {journal} {\bibinfo  {journal} {Physics Letters A}\ }\textbf
  {\bibinfo {volume} {223}},\ \bibinfo {pages} {337} (\bibinfo {year}
  {1996}{\natexlab{b}})}\BibitemShut {NoStop}%
\bibitem [{\citenamefont {Acevedo}\ and\ \citenamefont {Gobron}(2006)}]{AG06}%
  \BibitemOpen
  \bibfield  {author} {\bibinfo {author} {\bibfnamefont {O.~L.}\ \bibnamefont
  {Acevedo}}\ and\ \bibinfo {author} {\bibfnamefont {T.}~\bibnamefont
  {Gobron}},\ }\href {http://stacks.iop.org/0305-4470/39/i=3/a=011} {\bibfield
  {journal} {\bibinfo  {journal} {Journal of Physics A: Mathematical and
  General}\ }\textbf {\bibinfo {volume} {39}},\ \bibinfo {pages} {585}
  (\bibinfo {year} {2006})}\BibitemShut {NoStop}%
\bibitem [{\citenamefont {Kempf}\ and\ \citenamefont {Portugal}(2009)}]{KP09}%
  \BibitemOpen
  \bibfield  {author} {\bibinfo {author} {\bibfnamefont {A.}~\bibnamefont
  {Kempf}}\ and\ \bibinfo {author} {\bibfnamefont {R.}~\bibnamefont
  {Portugal}},\ }\href {\doibase 10.1103/PhysRevA.79.052317} {\bibfield
  {journal} {\bibinfo  {journal} {Phys. Rev. A}\ }\textbf {\bibinfo {volume}
  {79}},\ \bibinfo {pages} {052317} (\bibinfo {year} {2009})}\BibitemShut
  {NoStop}%
\bibitem [{Note1()}]{Note1}%
  \BibitemOpen
  \bibinfo {note} {The motivation of keeping all matrices nonvanishing
  originates in Ref. \cite {DP14} from the logic of deriving the graph
  inversely from a set of nonnull matrices. This is relavant from a derivation
  of the QW (more generally quantum automaton) from general topological
  principles of a countable set of interacting systems.}\BibitemShut {Stop}%
\bibitem [{Note2()}]{Note2}%
  \BibitemOpen
  \bibinfo {note} {Also in the finite Abelian case one can decompose the right
  regular representations into irreducible representations: the $k$-space is
  discrete and the diagonalization reads $C_l = \DOTSB \sum@ \slimits@
  _{j=1}^{i_l} \left | j \right >\left < j \right | e^{2\pi i\left ( \protect
  \frac {j}{i_l} \right )}$, where the wave-vectors are the $2\pi
  j/i_l$.}\BibitemShut {Stop}%
\end{thebibliography}%

\end{document}